%% file: main.tex
\documentclass[nonacm]{aamas}

\usepackage{balance} 

\usepackage{subfigure}
\usepackage{framed}

\setcopyright{ifaamas}
\acmConference[AAMAS '23]{Proc.\@ of the 22nd International Conference
on Autonomous Agents and Multiagent Systems (AAMAS 2023)}{May 29 -- June 2, 2023}
{London, United Kingdom}{A.~Ricci, W.~Yeoh, N.~Agmon, B.~An (eds.)}
\copyrightyear{2023}
\acmYear{2023}
\acmDOI{}
\acmPrice{}
\acmISBN{}

\acmSubmissionID{750}

\title[AAMAS-2023 Formatting Instructions]{Distributed Mechanism Design in Social Networks}

\author{Haoxin Liu}
\affiliation{
  \institution{ShanghaiTech University}
  \city{Shanghai}
  \country{China}}
\email{liuhx@shanghaitech.edu.cn}

\author{Yao Zhang}
\affiliation{
  \institution{ShanghaiTech University}
  \city{Shanghai}
  \country{China}}
\email{zhangyao1@shanghaitech.edu.cn}

\author{Dengji Zhao}
\affiliation{
  \institution{ShanghaiTech University}
  \city{Shanghai}
  \country{China}}
\email{zhaodj@shanghaitech.edu.cn}

\begin{abstract}
Designing auctions to incentivize buyers to invite new buyers via their social connections is a new trend in mechanism design~\cite{zhao22mechanism}. The challenge is that buyers are competitors and we need to design proper incentives for them to invite each other. For selling a single item, many interesting mechanisms have been proposed. However, all the mechanisms require the seller or a third party to be trustworthy to execute the mechanisms. In addition, the owner of the mechanism will know all the connections of the network after the execution, which poses a potential privacy issue. Hence, distributed mechanisms to avoid the privacy issue are more appealing in practice. Therefore, in this paper, we propose the first distributed mechanism in social networks without revealing buyers' private connections to anyone, and it achieves complete decentralization that does not rely on any trustworthy third party. Moreover, the centralized reduction of our mechanism also offers a novel way to compute players' contributions compared to the existing solutions. 
\end{abstract}


\newcommand{\BibTeX}{\rm B\kern-.05em{\sc i\kern-.025em b}\kern-.08em\TeX}

\begin{document}


\pagestyle{fancy}
\fancyhead{}


\maketitle 


\section{Introduction}
Mechanism design over social networks has recently attracted much attention from the researchers in AI~\cite{zhao22mechanism,li2022diffusion}. The design takes participants' connections into consideration and utilizes their connections to attract more participants, which works as a way to promote the mechanism to more participants via the participants' invitations. The challenge is that the participants may compete with each other in the game. For example, in auction, they compete for the limited resources, and in matching, they compete for the same preferred match. The existing mechanisms for the traditional settings 
cannot avoid the competition to incentivize participants to invite new participants. Therefore, new mechanisms are demanded in the new setting to attract more participants. We have seen a great progress in the directions of auctions, matching and coalitional games recently~\cite{DBLP:conf/ijcai/GuoH21,kawasaki2021mechanism,Yao2022coalitions}.


In this paper, we continue the study on auctions for selling a single item in the network setting. 
The existing mechanisms are centralized mechanisms which require the seller or a trusted center to execute the mechanism. 
However, after the execution, the center will know all participants' private connections, which will pose a potential privacy issue. 
Moreover, the seller can be anyone on the network which is not necessarily trustworthy and the mechanism design on networks aims to attract new participants without any third party. 
Therefore, to be more appealing in practice,
our goal is to design a distributed auction on a network which can be executed without a trusted center,
and does not reveal participants' private connections to anyone. 
Additionally, the network is distributed and is not owned by any single entity, which is also a natural environment for distributed mechanisms.

Different from centralized mechanisms, distributed mechanisms distribute the execution of the mechanisms to all participants. 
Therefore, the participants need to do more other than just reporting their private information/type as required in a centralized mechanism. This will give participants a larger action space and it becomes more challenging to prevent their manipulations. 

In the distributed mechanism design literature, researchers have tried to implement centralized mechanisms such as VCG in a distributed manner~\cite{parkes2004distributed}. 
However, the existing implementation still requires a trusted entity to do minimal computation or oversee the execution. 
Similarly, for the existing centralized mechanisms of selling a single item in social networks, it is also very hard to implement them in a completely distributed manner since they need the information about the graph structure like cut-points~\cite{li2017mechanism,li2022diffusion}. 

Against this background, we propose the first distributed mechanism for selling a single item in networks without relying on any trusted entity to oversee the computation. 
Moreover, our mechanism also provides a novel way to design incentives for the participants to invite others. Intuitively, our design rewards more buyers with a more decent reason than the existing methods, which gives them stronger incentives to participate in the mechanism. The reward of each buyer is computed according to her ability to connect the seller to the winner and also her ability to bring more valuable buyers. In summary, our contributions advance the state of the art in the following ways:
\begin{itemize}
    \item We propose the very first distributed mechanism for selling a single item in social networks, which does not require the assistant from a trusted center. 
    \item Our distributed mechanism also proposes a novel way to design the incentive for buyers to invite others, which gives more buyers positive rewards. This will incentivize buyers to participate.
\end{itemize}


\section{Related Work}\label{sec:relatedWork}
\textbf{Auctions in soical networks}. The first auction mechanism in social networks that can incentivize buyers to invite their neighbors is called Information Diffusion Mechanism (IDM)~\cite{li2017mechanism}. The main idea of IDM is to give compensation to cut-points of the highest bidder. Based on IDM, \citeauthor{zhao2018gidm} \cite{zhao2018gidm} extend it to homogeneous multi-item auctions, where each buyer only requires one item. Later, \citeauthor{li2022diffusion} \cite{li2022diffusion} characterize the necessary and sufficient conditions of incentive compatibility for all single-item auctions in social networks. Other related work along this research line can be found in a survey~\cite{DBLP:conf/ijcai/GuoH21}.

\noindent\textbf{Invitation incentives in other settings}. The idea that recruiting more participants by invitation has also been applied to many other game-theory settings. For example, \citeauthor{kawasaki2021mechanism} \cite{kawasaki2021mechanism} and \citeauthor{makoto2022matching} \cite{makoto2022matching} propose methods to incentivize invitation in matching market, and \citeauthor{Yao2022coalitions} \cite{Yao2022coalitions} initiate the model for invitation incentives in cooperative games. An overview about the problem in all these settings is given by \citeauthor{zhao2021bluesky}~\cite{zhao2021bluesky,zhao22mechanism}.

\noindent\textbf{Distributed mechanism design}. There also exists a rich literature about distributed mechanism design. \citeauthor{MondererT99} \cite{MondererT99} initialize the study on a simple single-item distributed auction problem where agents must forward messages from other agents to a center. 
Then, Feigenbaum et al. \cite{FeigenbaumPS01,FeigenbaumPSS02} firstly introduce the concept of the distributed algorithmic mechanism design. 
Following them, \citeauthor{parkes2004distributed} \cite{parkes2004distributed} and \citeauthor{Petcu06mdpop:faithful} \cite{Petcu06mdpop:faithful} have studied the distributed implementation of the VCG mechanism by proposing some principles to guide the distribution of computation. For other settings, \citeauthor{ShneidmanP03} \cite{ShneidmanP03} study the distributed implementation for interdomain routing. In our paper, to model a distributed mechanism in the new setting, we mainly follow the inspiration from the formal specification for distributed mechanism design introduced in~\cite{shneidman2004specification}. 
Some other related works have been summarized in \cite{feigenbaum2007distributed, feigenbaum2002distributed}.

However, all the above distributed mechanisms rely on a third party to verify some of the buyers' operations, while our distributed mechanism achieves complete decentralization that does not rely on any trustworthy third party.


\input{model.tex}

\input{TheMechanism.tex}

\input{TheoreticalAnalysis.tex}



\begin{acks}
This work is supported by Science and Technology Commission of Shanghai Municipality (No. 23010503000 and No. 22ZR1442200), and Shanghai Frontiers Science Center of Human-centered Artificial Intelligence (ShangHAI).
\end{acks}



\bibliographystyle{ACM-Reference-Format} 
\bibliography{sample}


\end{document}

%% file: model.tex
\section{The Model}\label{sec:model}
Consider the scenario where a seller $S$ sells one item in a social network, but she can only communicate with some of the buyers in the network. We model the network as an undirected graph $G=(V, E)$, where $V=N\cup\{S\}$ represents the set of all nodes in the network and the edge set $E$ contains all the connections among the nodes. 
The node set $N=\{1,2,...,n\}$ contains all potential buyers, and each buyer $i\in N$ has a private valuation $v_i\geq 0$ of receiving the item and a set of neighbors $r_i\subseteq V\backslash\{i\}$, where $j\in r_i$ if there is an edge between $i$ and $j$ in $E$. 
Particularly, denote $r_S$ as the seller's neighbor set.
We assume that a node can only directly communicate with her neighbors, which is also a feature of a modern social network. 
Formally, let $\theta_i=(v_i, r_i)$ be the private type of buyer $i\in N$ and $\Theta_i=\mathbb{R}_{\geq0}\times\mathcal{P}(V)$ be the type space of $i$, where $\mathcal{P}(V)$ is the power set of $V$. 
Let the joint vector $\theta=(\theta_1, \theta_2, ..., \theta_n)$ denote the type profile of all buyers and $\Theta=\Theta_1\times\Theta_2\times\cdots\Theta_n$ be the type profile space. 
Denote $\theta_{-i}$ as the type profile of all buyers except $i$ and $\theta$ can also be written as $(\theta_i, \theta_{-i})$.

In this model, we assume that initially only the seller's neighbors $r_S$ are aware of the sale as she cannot inform the others by herself. The goal is to incentivize the informed buyers to use their connections to invite more buyers to join the sale. A buyer's invitation is modeled by reporting her neighbors here. Invited buyers can further invite other buyers, and eventually only the buyers who are invited can join the sale. Hence, we have to determine who are valid buyers according to their reported neighbors. 
\begin{definition}\label{def:valid}
Given the buyers' reported type profile $\theta'$, for each buyer $i\in N$ with reported type $\theta'_i=(v'_i, r'_i)$, build an edge between $i$ and $j$ if $j\in r'_i$. We say a buyer $i$ is \textbf{valid} if there exist a path connecting $i$ with the seller. Denote the subgraph containing all valid buyers as $G(\theta')$.
\end{definition}

In principle, buyers who are not invited are not aware of the sale and will not report anything. However, to make the definitions clean, we assume all buyers report in the model, but only valid buyers are considered in the sale. 
Given the above setting, our goal is to design distributed mechanisms.
To make the definitions easy to follow, we first define the centralized mechanisms. 

\begin{definition}\label{def:centralizedMechanism}
A centralized (direct-revelation) mechanism in social networks is a 2-tuple $M=(\pi, p)$, where $\pi=\{\pi_i\}_{i\in N}$ is the allocation function and $p=\{p_i\}_{i\in N}$ is the payment function of all buyers. Particularly, $\pi_i: \Theta\rightarrow[0,1]$ and $p_i:\Theta\rightarrow \mathbb{R}$ are the allocation and payment functions for $i$ respectively, and they further 
satisfy that for all reported type profile $\theta'\in \Theta$, (1) for all invalid buyers $i\in N\backslash G(\theta')$, $\pi_i(\theta')=0$ and $p_i(\theta')=0$, and (2) for all valid buyers $i\in G(\theta')$, $\pi_i(\theta')$ and $p_i(\theta')$ are independent of the reports of the invalid buyers.
\end{definition}

Given the buyers’ reported type profile $\theta'$, $\pi_i(\theta')$ represents the probability for allocating the item to buyer $i$.
Given a mechanism $M$, a reported type profile $\theta'$, the utility of a buyer $i$ of type $\theta_i$ is defined as $u_i(\theta_i, \theta', (\pi, p))=\pi_i(\theta')\cdot v_i-p_i(\theta').$ 
In the centralized scenarios, we say a mechanism $M$ is \textit{incentive compatible} if truthfully revealing the type is a buyer's dominant strategy no matter what the others report.
\begin{definition}\label{def:ic}
A centralized mechanism $M=(\pi, p)$ is \textbf{incentive compatible} (IC) if for all $i\in N$, all $\theta'_i\in\Theta_i$ and all $\theta'_{-i}\in \Theta_{-i}$,
$$u_i(\theta_i, (\theta_i,\theta'_{-i}), (\pi, p))\geq u_i(\theta_i, (\theta'_i,\theta'_{-i}), (\pi, p)).$$
\end{definition}

Another desirable property is individual rationality, which guarantees that a buyer will not suffer a loss in the mechanism as long as she truthfully reports her type.
\begin{definition}\label{def:ir}
A centralized mechanism $M=(\pi, p)$ is \textbf{individually rational} (IR) if for all buyers $i\in N$, all $\theta_i\subseteq \Theta_i$, and $\theta'_{-i}\in \Theta_{-i}$,
$$u_i(\theta_i, (\theta_i,\theta'_{-i}), (\pi, p))\geq0.$$
\end{definition}


In centralized scenarios, the only action a buyer needs to do is reporting her type to the center, which is the only space for manipulation. 
Different from centralized mechanisms which are executed by a trusted center, distributed mechanisms distribute the execution to all participants.
This will enlarge the action space of the buyers and cause more possibilities to manipulate.
Therefore, in a distributed mechanism, we also need to guarantee that the buyers execute the mechanism correctly.
Hence, it is necessary to introduce the concept about \textit{strategy} to capture how a buyer behaves in all states of the mechanism.
Let $s_i$ denote the strategy of buyer $i$ which is parameterized by $i$'s type $\theta_i$ and let $\Sigma_i$ be $i$'s strategy space, which includes all strategies $i$ can perform. 
Let $s(\theta)=(s_1(\theta_1), s_2(\theta_2), ..., s_n(\theta_n))$ be a strategy profile of all buyers under type profile $\theta$ and let $s_{-i}(\theta)=(s_j(\theta_j))_{j\neq i,j\in N}$. 

\begin{definition}\label{def:DM}
A distributed mechanism $d^M$ is a tuple $d^M=(\Sigma, (\pi, p), s^M)$, where $\Sigma=(\Sigma_1, ..., \Sigma_n)$ is the strategy space of all buyers, $\pi=\{\pi_i\}_{i\in N}$ is the allocation function, $p=\{p_i\}_{i\in N}$ is the payment function, and $s^M=(s^M_1, ..., s^M_n)\in \Sigma$ is the intended strategy of the mechanism. Particularly, $\pi_i:\Sigma\rightarrow\{0,1\}$ and $p_i:\Sigma\rightarrow \mathbb{R}$ are the allocation and payment functions for $i$ respectively.
\end{definition}

For every buyer $i$,  the intended strategy $s^M_i\in\Sigma_i$ can be considered as a series of algorithms or actions that the mechanism requires $i$ to perform. $s^M_i$ is parameterized by the private type $\theta_i$ of buyer $i$, and $s^M_i(\theta_i)$ indicates which actions buyer $i$ should execute in every state of the mechanism. 
In the centralized scenario, the strategy of each buyer is only reporting her type to the center, so the strategy space is reduced to type space, i.e. $\Sigma_i=\Theta_i$. Since there is only one kind of action, which is private information revelation, we can define $s_i(\theta_i)=\theta'_i$ by viewing $s_i$ as a mapping function from her type $\theta_i$ to the type space $\Theta_i$, and we usually intend each buyer to truthfully report her type, i.e., $s^M_i(\theta_i)=\theta_i$. However, in the distributed scenario, the strategy space is very complex and does not have a standard structure, which includes many varieties of actions besides reporting type. Here we refer to the canonical literature \cite{shneidman2004specification} and decompose the strategy $s_i$ into three kinds of actions, $s_i =(t_i,q_i,f_i)$, which are information-revelation action $t_i$, message-passing action $q_i$, and computational action $f_i$. For each buyer $i\in N$, $t_i$ decides whether to reveal her type truthfully, $q_i$ determines how she passes messages to her neighbors (for example, she can decide whether to deliver to one neighbor or multiple neighbors), and $f_i$ decides how to conduct local computation based on messages she receives.
Similarly, the intended strategy $s^M_i$ can also be represented as $(t^M_i, q^M_i, f^M_i)$.

Instead of using $\pi(\theta')$ and $p(\theta')$ to represent the outcomes that depend only on the reported information, we now must update the notation to $\pi(s(\theta))$ and $p(s(\theta))$ that depend on the sequence of actions taken by buyers. 
Hence, the utility of a buyer is updated to $u_i(\theta_i, s(\theta), (\pi, p))=\pi_i(s(\theta))\cdot v_i-p_i(s(\theta))$. 
We have mentioned that the challenges in distributed mechanism design are different from centralized mechanism design, because the computation of a distributed mechanism is performed by the strategic buyers in the absence of a trusted center. The buyers can manipulate the computation to their own interests. In such a scenario, the pursuit of IC might be impossible, because there might be no single computational behavior that is optimal regardless of what the other buyers do~\cite{feigenbaum2007distributed}. Hence, we will focus on a more suitable solution concept called \textit{ex-post incentive compatibility}, which can be viewed as a compromise of distributing the computation to the buyers.
\begin{definition}\label{def:ex-post}
A distributed mechanism $d^M=(\Sigma, (\pi, p), s^M)$ is \textbf{ex-post incentive compatible} if for all $\theta\in\Theta$, all buyers $i\in N$, and all $s_i\in\Sigma_i$,
$u_i(\theta_i, (s^M_i, s^M_{-i}), (\pi, p))\geq u_i(\theta_i, (s_i, s^M_{-i}), (\pi, p)).$
\end{definition}

It means no one can obtain a higher utility by deviating from the equilibrium that everyone executes the intended strategy. If a mechanism is ex-post IC, then $s^M$ is an ex-post Nash equilibrium. Although weaker than a dominant strategy equilibrium, ex-post IC is also a strong solution concept because it does not require buyers to have any knowledge of the private types of the others.

In addition, for any buyer $i\in N$, if we restrict her strategy to $s_i(\theta_i)=(t_i, q^M_i, f^M_i)$, there always exists a centralized mechanism $M$ such that $E[\pi(s(\theta'))]=\pi'(\theta')$ and $E[p(s(\theta'))]=p'(\theta')$\footnote{We take the expectation results of the distributed mechanism because distributed mechanisms usually have randomized outcomes in practice and the expectation results are only used for analysis.}, where $\pi$ is the outcome of the distributed mechanism $d^M$, $\pi'$ is the outcome of the corresponding centralized mechanism $M$, $p$ and $p'$ are the payment functions of these two mechanisms respectively.
We call this mechanism $M$ the \emph{centralized reduction mechanism} (CRM) of $d^M$ and say $d^M$ is a \emph{distributed implementation} of $M$. 

%% file: TheMechanism.tex
\section{The Mechanism}\label{sec:mechanism}

In this section, we will formally describe the very first distributed mechanism in social networks called the \emph{Sequential Resale Auction} (SRA). There already exist many centralized mechanisms in social networks such as the Information Diffusion Mechanism \cite{li2017mechanism, li2022diffusion}. These mechanisms are highly dependent on the cut-points of the network and they only compute incentives for the cut-points which do not form complete paths. 
In addition, it is hard to locate the cut-points in decentralized settings since no one can know the structure of the whole graph, 
and we cannot pass the item distributively without a complete path. 
Therefore, our distributed mechanism will not implement the existing centralized diffusion auctions. 
Moreover, the centralized reduction of our mechanism gives another novel way to design the diffusion incentive which is based on their connection power to the item receiver. 

\subsection{Sequential Resale Auction}
We describe the distributed 
auction as a three-stage process and buyers will perform different kinds of actions in each stage.
In the first stage, the buyers diffuse the sale information to their neighbors.
In the second stage, the buyers collect their invited neighbors' bids and represent them to join the sale.
In the last stage, we do sequential resales from the seller to the final winner.

\noindent \textbf{Stage 1 (Top-down Diffusion)}: The first stage is \textit{top-down diffusion}, in which the sale information is spread in the social network starting from the original seller $S$. Any buyer $i\in N$ who is aware of the sale can decide her information-revelation action in this stage.

\begin{definition}[Information-revelation Action]\label{def:information-revelation}
Given buyer $i$'s type $\theta_i=(v_i, r_i)$, her information-revelation action $t_i$ is to decide her bid $v'_i$ in the sale and choose neighbors $r'_i\subseteq r_i$ to invite, i.e., $t_i=(v'_i, r'_i)$. The intended information-revelation action is to truthfully reveal her type, i.e., $t^M_i=(v_i,r_i)$.
\end{definition}

The intended information-revelation action is the same as the reporting action in centralized mechanisms.
However, instead of reporting her valuation and neighbors to the seller in centralized scenarios, a buyer now only needs to invite her neighbors on her own interest and does not need to tell her valuation to anyone.
When buyer $i$ invites a neighbor $j\in r'_i$ to join the sale, the edge $e_{ij}$ becomes a directed edge from $i$ to $j$ and we say buyer $i$ is an \textit{inviter} of buyer $j$. A buyer may have multiple inviters on a network and she can further invite her neighbors except for her inviters. 
Finally, the social network becomes a connected directed graph $G'$ containing all valid buyers.
Note that the graph is unknown to any agent, and everyone only knows who invites her and who she invites. 

\noindent \textbf{Stage 2 (Bottom-up Aggregation)}: The second stage is \textit{bottom-up aggregation}, where each buyer determines her message-passing action $q_i$ and computational action $f_i$. 
In our distributed mechanism, each buyer may receive several messages from the neighbors she invites, and she can aggregate those messages into a new message called \emph{aggregated bid} and pass it to her inviters.
We first discuss each buyer's computational action $f_i$ of the aggregation process and describe the message-passing action later.
For simplicity, we denote the aggregated bid as $b_i\in\mathbb{R}_{\geq0}$ and use $B_i$ to represent the set of all bids buyer $i$ receives.
The computational action $f_i$ corresponds to an aggregation algorithm, which takes $B_i$ and her own bid $v'_i$ as input and outputs a new bid, and we denote the computational action space as $F$ to contain all possible aggregation algorithms. 

\begin{definition}[Computational Action]\label{def:compuatation}
Given a buyer $i$'s received bids set $B_i$, and her bid $v'_i$, the computational action $f_i\in F$ will generate her aggregated bid $b_i=f_i(B_i, v'_i)$. The intended computational action $f^M_i$ is to select the largest bid among all the bids collected by buyer $i$ as her aggregated bid, i.e., $f^M_i(B_i,v'_i)=\max(B_i\cup\{v'_i\})$.
\end{definition}

After generating $b_i$, the next action a buyer can manipulate is to choose whether to pass the message $b_i$ truthfully and which inviters to pass the message to. 
Since misreporting $b_i$ to $b'_i$ is actually the same as choosing another aggregation algorithm whose output is the misreported value $b'_i$ and reporting $b'_i$ truthfully, we categorize this kind of manipulation into $f_i$ and assume buyers will truthfully report their aggregated bids when considering message-passing action.
Hence, the message-passing action only cares about reporting the aggregated bid to one or more inviters, and we denote buyers' message-passing action space as $Q$ to contain all possible actions. 

\begin{definition}[Message-passing Action]\label{def:message-passing}
Buyer $i$'s message-passing action $q_i\in Q$ is to select one or more inviters to report her aggregated bid. The intended message-passing action $q^M_i\in Q$ is to randomly select one inviter. 
\end{definition}

The intended message-passing action is designed as above because all inviters are equivalent from a buyer's local view, and it prevents the buyer's bid from being aggregated multiple times. 
For an inviter who receives buyer $i$'s aggregated bid, the inviter still does not know which buyer the bid actually belongs to, and this can effectively protect buyers' privacy.
The second stage starts from the leaf nodes, who do not invite anyone and just report their bids, and ends until the original seller receives all her neighbors’ aggregated bids.
If all buyers execute the intended message-passing action, the social network will finally become a directed tree.
An example social network after this stage is shown in Figure~\ref{fig:stage2}. 

\begin{figure}[thbp]
\centering
\subfigure[]{
    \includegraphics[width=0.25\linewidth]{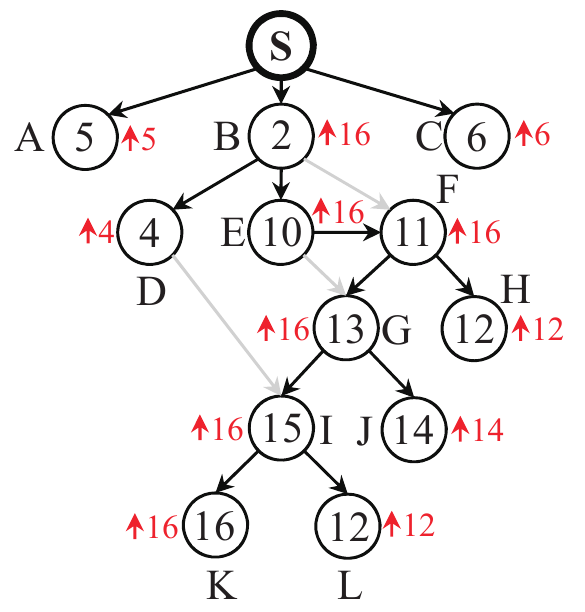}
    \label{fig:stage2}
    }
\subfigure[]{
    \includegraphics[width=0.25\linewidth]{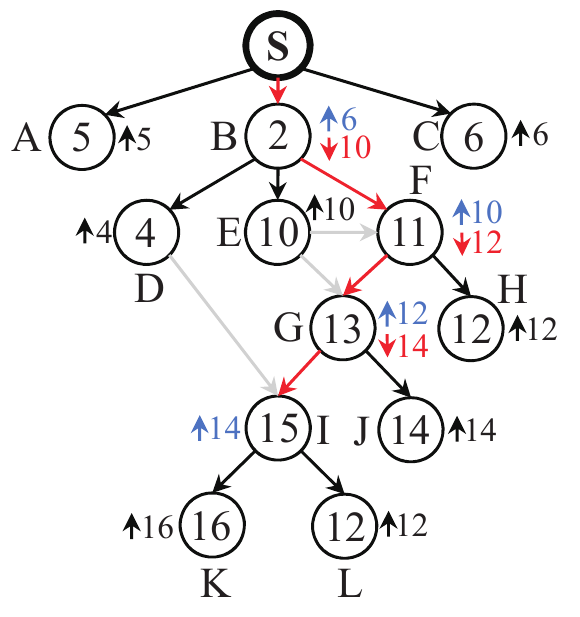}
    \label{fig:stage3}
    }
\caption{(a) An example of social network after bottom-up aggregation stage. The red number beside each node is the computation result using the intended algorithm. (b) An example of social network after top-down allocation stage. The blue number beside each node is the purchasing price to win the item, and the red number is the selling price.}
\label{fig:runningExample}
\end{figure}

\noindent \textbf{Stage 3 (Top-down Allocation)}: The third stage is \textit{top-down allocation}. We model the entire auction as a series of resales, and each resale is called a \emph{local auction}.
Suppose the item is resold to buyer $i$ currently, $i$ can initiate a local auction and notify all neighbors she invites in Stage 1 to join (including neighbors who do not pass messages to her in Stage 2). 
Notice that if a buyer $j$ in this local auction does not pass her bid to $i$ in Stage 2, 
then there must exist another path connecting the $i$ and $j$.
To prevent other buyers on the path from using the buyer $j$'s bid to compete with $j$,
$j$ should connect herself only to the current seller $i$ and disconnect from other inviters.
All participants affected by this disconnection should re-aggregate their bids.
For example, in Figure~\ref{fig:stage2}, suppose $B$ holds a local auction now. $F$ is invited by $B$, so $F$ should disconnect herself from $E$ and join $B$'s auction, such that $E$'s aggregated bid will become 10 since she cannot receive message from $F$ now. We do this because both $E$ and $F$ are the participants in the local auction held by $B$, then $E$ should not compete for the item with $F$'s bid. 
The definition of a local auction hosted by buyer $i$ is given as follows:
\begin{definition}\label{def:local-auction}
A local auction hosted by buyer $i$, $\hat{M}_i$, is composed of a local allocation function $\hat{\pi}^i=\{\hat{\pi}^i_j\}_{j\in r'_i}$ and a local payment function $\hat{p}^i=\{\hat{p}^i_j\}_{j\in r'_i}$, where $\hat{\pi}^i_j\in\{0,1\}$ and $\hat{p}^i_j\in \mathbb{R}$ are the local allocation and local payment for participant $j$ respectively. 
\end{definition}

The seller $i$ determines the local allocation results $\hat{\pi}^i$ and local payment results $\hat{p}^i$ based on her received bids set $B_i$, her own bid $v'_i$, and her purchasing price $\bar{p}_{i}$ which represents the amount she pays to win the item from the previous local auction\footnote{The purchasing price of the original seller is initialized as $0$.}. 
The local auction is a classic one-layer VCG auction with a reserve price which is equal to the purchasing price. 
The local auctions run iteratively from the original seller until someone chooses to keep the item, where all the local sellers form a resale path.
Note that a local seller actually participates in only two consecutive local auctions, once as a buyer and once as a seller. 
A running example of this stage is shown in Figure~\ref{fig:stage3}.

To prevent malicious manipulation on the purchasing price, we utilize the Distributed Ledger Technology (DLT)~\cite{sunyaev2020distributed} to encrypt the record of each purchasing price, and everyone can access the ledger to know the historical transactions. 
In terms of how to use a DLT, we may face three different situations:
(1) The exact DLT is a common knowledge, e.g., in a completely distributed environment like the digital world (the metaverse), all exchanges are recorded in a kind of DLT, where they may have a fixed DLT for all kinds of exchanges. In this case, we don’t need to propagate the information of the DLT.
(2) The need for a DLT is common, but which DLT is not common, e.g., in the digital world, we use multiple cryptocurrencies to exchange goods, and each market may have different choices. In this case, the agents of our mechanism also need to propagate the DLT information. Here, each agent doesn’t have any incentive to misreport the DLT, because the invitees will easily notice the manipulation because they cannot find the exchange records of the previous resales in a misreported DLT.
(3) The last case is that the network is owned by a centralized platform and DLT is not common to record exchanges. Then, an agent may not tell the DLT information to her neighbors at all so that she can increase her reserve price to gain more. To prevent such manipulation, we can make the rules of the mechanism public and the agents only join a mechanism where they believe no one can manipulate. Then each agent in our mechanism will also need to propagate the mechanism id where misreporting the id is not beneficial.

The above distributed mechanism is summarized as follows:
\begin{framed}
 \label{frame:DDAM}
 \noindent\textbf{Sequential Resale Auction (SRA)}
 
 \noindent\rule{\textwidth}{0.5pt}
\noindent (1) \textbf{Top-down diffusion}. The original seller $S$ starts to propagate the auction information to her neighbors. Each buyer who is aware of the sale executes her information-revelation action $t_i=(v'_i,r'_i)$, which contains both deciding her bid $v'_i$ and inviting her neighbors $r'_i\subseteq r_i$. The social network finally becomes a directed graph $G'$.
    
\noindent (2) \textbf{Bottom-up aggregation}. Each buyer executes the message-passing action $q_i$ to pass the result of the computational action $b_i=f_i(B_i,v'_i)$ to one of her inviters. This stage starts from all leaf node who invites no buyers, and terminates until the original seller receives all her neighbors' aggregated bids.

\noindent (3) \textbf{Top-down allocation}. The original seller $S$ starts the first local auction. For a local auction hosted by $i$, let $b^{1\text{st}}_{r'_i}$ and $b^{2\text{nd}}_{r'_i}$ be the highest and second-highest bid among all participants respectively. The local allocation and payment function of a participant $j\in r'_{i}$ are defined as below:
     \begin{itemize}
         \item[$\bullet$] \textbf{Local allocation function}:
         \begin{equation}\label{eq:allocation}
         \hat{\pi}^{i}_{j}=\left\{
            \begin{array}{lll}
            1,\quad \textrm{if}~v'_j=b^{1\text{st}}_{r'_i}, \text{ and } v_i<\max\{\bar{p}_{i}, b^{2\text{nd}}_{r'_i}\},\\
            0,\quad \textrm{otherwise}.
            \end{array}
        \right.
        \end{equation}
        \item[$\bullet$] \textbf{Local payment function}:
        \begin{equation}\label{eq:payment}
        \hat{p}^{i}_{j}=\left\{
            \begin{array}{ll}
            \max\{\bar{p}_{i}, b^{2\text{nd}}_{r'_i}\}, \quad &\textrm{if}~\hat{\pi}^{i}_{j}=1,\\
            0, \quad &\textrm{otherwise}.
            \end{array}
        \right.
        \end{equation}
     \end{itemize}
     After the local auction $\hat{M}_i$ finishes, the local winner $w$, i.e., $\hat{\pi}^{i}_{w}=1$, will host next local auction $\hat{M}_w$ and her payment in $\hat{M}_i$ will be the purchasing price $\bar{p}_{w}$ in the new local auction $\hat{M}_w$. The whole resale process terminates if a local seller $i$ keeps the item, i.e., $\Sigma_{j\in r'_i}\hat{\pi}_j^i=0$.
\end{framed}

In Equation~\eqref{eq:allocation}, the local seller will not allocate the item to anyone if her valuation is no less than the selling price, i.e., $\Sigma_{j\in r'_i}\hat{\pi}_j^i=0$ if $v_{i}\geq \max\{\bar{p}_{i}, b^{2\text{nd}}_{r'_i}\}$. 
In Equation~\eqref{eq:payment}, the local winner $w$ (if exists) pays the selling price, $\max\{\bar{p}_{i}, b^{2\text{nd}}_{r'_i}\}$, to the local seller $i$ and other buyers' payments are zero. 
Therefore, for a local seller $i$ except for the winner, the transactions she involves in include purchasing the item from previous local auction by paying $\bar{p}_{i}$, then selling the item to next local seller by charging $\max\{\bar{p}_{i}, b^{2\text{nd}}_{r'_i}\}$, so her utility can be represented as $u_i=\max\{\bar{p}_{i}, b^{2\text{nd}}_{r'_i}\}-\bar{p}_{i}$. 
For the winner $w$ who keeps the item, her utility can be represented as $u_w=v_w-\bar{p}_{w}$.

\subsection{Centralized Reduction of the SRA}\label{sec:CRM}
To facilitate a better understanding of our distributed mechanism, we will briefly discuss the centralized reduction mechanism of the sequential resale auction in this part.
The centralized setting is equivalent to the situation that everyone follows the intended message-passing action and intended computational action, and she can only manipulate how to reveal her private type in her information-revelation action $t_i=(v'_i,r'_i)$.
Therefore, everyone's strategy $s_i(\theta_i)=(t_i, q^M_i, f^M_i)$ is reduced to the reported type $\theta_i'$ which is consistent with traditional centralized mechanisms.

Reviewing the first stage of the SRA, each buyer does not need to report her valuation to anyone and only needs to invite her neighbors to diffuse the auction information. In the centralized setting, each buyer is required to report her type to the seller, so the seller can get access to the whole network structure $G(\theta')$ containing all valid buyers. 
Hence, in the centralized scenarios, the seller can quickly locate the highest bidder denoted as $z$ with $v'_z=v^{1\text{st}}_{G(\theta')}$, where we denote $v^{1\text{st}}_{\mathcal{D}}=\max_{i\in \mathcal{D}}v'_i$ to be the highest reported valuation in the subset $\mathcal{D}\subseteq G(\theta')$.



The social network becomes a randomized tree after the second stage in our distributed auction, we also generate a spanning tree $T$ randomly from 
$G(\theta')$ in the centralized reduction mechanism. 
Denote the set containing all possible spanning trees as $\mathcal{T}$.
The seller can determine the simple path from $S$ to $z$ in each spanning tree. 
We define a special class of paths and explain how it corresponds to the resale path in the third stage of our distributed auction. 
\begin{definition}\label{def:diffusion_path}
A \textbf{diffusion path} to buyer $m$ is a simple path from $S$ to $m$, denoted as $h^m=(h_0, h_1, ..., h_k)$, where $h_0 = S$, $h_k = m$, and it satisfies for any two buyers $h_i$, $h_j$ ($i<j-1$), there is no edge between $h_i$ and $h_j$ in the connected graph $G(\theta')$. That is, there are no back-edges between any two buyers on a diffusion path.
\end{definition}

For each spanning tree, if the simple path from $S$ to $z$ is not a diffusion path, we make a transformation on the path. 
A transformation runs as follows:
for each existing back-edge between $h_i$, $h_j$ ($i<j-1$) on the path $h^z$, remove all nodes between $h_i$ and $h_j$ on the path and add the back-edge into the path.
As we described in the third stage of the SRA, for any local auction, each neighbor of the local seller needs to disconnect herself from other inviters and add the edge to the local seller. 
The transformation is actually corresponding to this operation we mentioned before. 
The intuition of the centralized mechanism is to resell the item iteratively on the diffusion path. 
The difference between the selling price and the purchasing price of each local seller is considered as her connecting contribution on the diffusion path.
One place where the mechanism is different from the SRA is that we can enumerate all spanning trees in the centralized scenarios since the seller is aware of the entire network structure, so we will average a buyer's connection contribution over all spanning trees as her final payoff. 

The centralized reduction mechanism is summarized as follows:
\begin{framed}
 \label{frame:SRA}
\noindent\textbf{Centralized Reduction Mechanism of the SRA}

\noindent\rule{\textwidth}{0.5pt}

\noindent (1) Given a reported type profile $\theta'\in \Theta$, build the subgraph $G(\theta')$ of valid buyers, and find the valid buyer $z$ with the highest valuation in $G(\theta')$ (with random tie-breaking).


\noindent (2) For a spanning tree $T$ generated from $G(\theta')$, check whether the simple path from $S$ to $z$ is a diffusion path. If it is, go to (4); else, go to (3).


\noindent (3) Transform the simple path to a diffusion path.


\noindent (4) Denote the corresponding diffusion path from $S$ to $z$ as $h^z=\{S, h_1, h_2, ..., z\}$ and use $T_{-i}$ to represent the remaining buyer set without the participation of $i$ on the tree. Allocate the item on the current spanning tree where the allocation function can be recursively defined as:
\begin{equation}\label{eq:crm_pi}
     \pi_i^T(\theta')=\left\{
     \begin{array}{lll}
          1 &\text{if } i=h_j\in h^z, v'_i=v^{1\text{st}}_{T_{-h_{j+1}}},\\
          &\text{and }\sum_{k\in T_{-i}}\pi^T_k(\theta')=0,\\
          0 &\text{otherwise.}
     \end{array}
     \right.
 \end{equation}

\noindent (5) Denote the winner as $w=h_l\in h^z$. The payment function for the current spanning tree is defined as:
\begin{equation}\label{eq:crm_p}
     p_i^T(\theta')=\left\{
     \begin{array}{lll}
          v^{1\text{st}}_{T_{-w}} &\text{if }i=w,\\
          v^{1\text{st}}_{T_{-h_j}}-v^{1\text{st}}_{T_{-h_{j+1}}} &\text{if } i=h_j\in h^z, j<l,\\
          0 &\text{otherwise.}
     \end{array}
     \right.
 \end{equation}

\noindent (6) For each possible spanning tree generated from $G(\theta')$, repeat (2)-(5). 
In the process, for each buyer $i\in G(\theta')$, count her total number of wins denoted as $cnt(i)$, and her total payments denoted as $sum(p^T_i)$.

\noindent (7) Denote the number of all possible spanning trees as $|\mathcal{T}|$. 
The overall allocation function and payment function are defined as:
\begin{equation}\label{eq:crm_all}
     \pi_i(\theta')=\frac{cnt(i)}{|\mathcal{T}|},\quad p_i(\theta')=\frac{sum(p^T_i)}{|\mathcal{T}|}
 \end{equation}

\end{framed}

Next, we formally analyze the relationship between the mechanism described above and the SRA mechanism.
\begin{theorem}\label{thm:CRM}
The mechanism above is the centralized reduction mechanism of the sequential resale auction.
\end{theorem}
\begin{proof}
Without loss of generality, we only need to compare the allocation and payment results of the two mechanisms on the same spanning tree $T$. In the following proof, we always use $h_j$ to represent $i$ on the path to $z$. 
We first prove that each local seller's purchasing price is the highest bid without her participation, i.e., $\bar{p}_{i}=v^{1\text{st}}_{T_{-h_{j}}}$ using the mathematical induction method.
Suppose the condition holds for $i$'s previous buyer, i.e., $\bar{p}_{h_{j-1}}=v^{1\text{st}}_{T_{-h_{j-1}}}$. Then, $\bar{p}_{h_j}=\max\{\bar{p}_{h_{j-1}}, b^{2\text{nd}}_{r'_{h_{j-1}}}\}$ where $b^{2\text{nd}}_{r'_{h_{j-1}}}=\max_{k\in r'_{h_{j-1}}\setminus h_{j}}b_k$, so $\bar{p}_{h_j}=v^{1\text{st}}_{T_{-h_{j}}}$ since $T_{-h_{j}}=T_{-h_{j-1}}\cup r'_{h_{j-1}}\setminus h_{j}$.
On this basis, it is easy to deduce that the allocation and payment results of the two mechanisms are the same on the same spanning tree.
\end{proof}

\noindent\emph{Discussion}. 
(1) Different from the existing centralized mechanisms like the Information Diffusion Mechanism (IDM)~\cite{li2017mechanism} which only gives rewards to the critical ancestors, our mechanism can reward more buyers including those non-cut-points, which gives all buyers stronger incentive to participate in the mechanism. 
(2) The seller's revenue is always no less than that of traditional VCG among neighbors without diffusion (see Proposition~\ref{thm:revenue}). 
(3) Most importantly, the mechanism presents a method to calculate the payoff of each buyer over complete paths, which provides a sound basis for its corresponding distributed implementation, and this is why we show this mechanism here.

%% file: TheoreticalAnalysis.tex
\section{Evaluations}\label{sec:theoreticalAnalysis}
In this section, we provide theoretical analysis for the sequential resale auction. We also conduct experiments to compare the centralized reduction mechanism of the SRA with the IDM, which is a representative of existing centralized mechanisms. 

\subsection{Theoretical Analysis}
In this part, 
without loss of generality, we consider one randomized instance of running the distributed auction in the following proofs since the properties of IC and IR hold for the whole mechanism if they hold in all instances.
Suppose $z$ is the highest valid buyer and $w$ is the winner in the instance. Let $\mathcal{Y}=\{S, y_1, y_2, ..., w\}$ represent the resale path containing all local sellers involved in the resale process. According to the payment function of the SRA defined in Eq.~\ref{eq:payment}, only buyers on the resale path are involved in money transactions and may gain nonnegative utilities. Therefore, we classify all buyers into three different categories: 
(1) The final winner: $w$. (2) Local sellers: $\mathcal{Y}\setminus \{S, w\}$, i.e., all buyers on the resale path except the original seller and final winner. (3) Other buyers: $\forall i\notin \mathcal{Y}$, i.e., all other buyers who are not on the resale path.

Given the above classification, we will prove that the sequential resale auction satisfies the properties of IR and ex-post IC. Before that, we first show that a buyer's payment is independent of her bid when all buyers execute the intended strategy.

\begin{lemma}\label{lem:independent}
When everyone executes the intended message-passing action $q_i^M$ and intended computational action $f_i^M$, each buyer $i$'s payment is independent of her bid.
\end{lemma}
\begin{proof}
As we mentioned before, the social network becomes a directed tree after the second stage if all buyers execute the intended message-passing action. 
Hence, the purchasing price $\bar{p}_{i}$ of each local seller $i$ must come from another branch she does not belong to, which is independent of her bid aggregated only in the branch she belongs to. Therefore, (1) for the winner $w$, her payment is $\bar{p}_{i}$ which is not dependent on her bid; (2) for other local seller $i\in\mathcal{Y}\setminus\{S,w\}$, her utility is $u_i=\max\{\bar{p}_{i}, b^{2\text{nd}}_{r_i}\}-\bar{p}_{i}$, where $b^{2\text{nd}}_{r_i}$ is the second-highest aggregated bid among her children, which is also independent of her bid; 
(3) for all other buyers, their payments are always 0.
\end{proof}


We then show that no buyer in the SRA will gain a negative utility as long as she uses her true valuation as her bid and everyone executes intended message-passing action and intended computational action.

\begin{theorem}\label{thm:ir}
The sequential resale auction is IR.
\end{theorem}
\begin{proof}
Assume that buyer $i$ reveals her valuation truthfully. 
(1) If she is the winner, $i=w$, her utility is $u_i=v_i-\bar{p}_{i}\geq0$ because she will choose to keep the item only if $v_i\geq\max\{\bar{p}_{i}, b^{2\text{nd}}_{r'_i}\}\geq \bar{p}_{i}$ according to Eq.~\ref{eq:allocation}.
(2) If she is a local seller, $i\in \mathcal{Y}\setminus\{S, w\}$, her utility is $u_i=\max\{\bar{p}_{i}, b^{2\text{nd}}_{r'_i}\}-\bar{p}_{i}\geq0$.
(3) For any other buyer, her utility is 0.
Therefore, the SRA is individually rational.
\end{proof}

Now we show that the intended strategy $s^M$ is an ex-post Nash equilibrium for all buyers, i.e., no one can gain a higher utility by deviating the intended strategy herself.

\begin{theorem}\label{thm:ex-postIC}
The sequential resale auction is ex-post IC. 
\end{theorem}
\begin{proof}
We discuss the effects of different actions by different categories of buyers separately.

1. For determined $q^M_i, f^M_i$ and other buyers' strategies $s^M_{-i}$, we first prove that disobeying the intended information-revelation action $t^M_i$ cannot gain a higher utility for any buyer. 
As the information-revelation action concerns the buyer's revealed valuation $v'_i$ and her invited neighbors $r'_i$, we first prove that, for all kinds of buyers, revealing the true valuation maximizes their utilities when their invited neighbors are also determined. 
Based on Lemma~\ref{lem:independent}, we only need to consider how a buyer's bid affects her category.

\noindent \textbf{Case 1.1}
For the winner $w$, any $v'_w\geq v_w$ will only increase the aggregated bid on the resale path. The allocation result will not change, so her utility remains unchanged.
For any bid $v'_i<v_i$, she may still be the winner with unchanged utility. If she becomes a local seller, her new utility will be $u'_i=\max\{\bar{p}_{i}, b^{2\text{nd}}_{r'_i}\}-\bar{p}_{i}\leq v_i-\bar{p}_{i}=u_i$ because she is the original winner with $v_i\geq \max\{\bar{p}_{i}, b^{2\text{nd}}_{r'_i}\}$. Otherwise, her utility will degenerate to 0. 

\noindent \textbf{Case 1.2}
For a local seller $i\in\mathcal{Y}\setminus\{S,w\}$, we can defer that her valuation is not larger than her aggregated bid since she is not the final winner.
Therefore, the message she passes to her inviter will not change for any $v'_i<v_i$, so the allocation result and her utility will not change.
If she reveals a larger valuation to become the new winner, her utility will be $u'_i=v_i-\bar{p}_{i}\leq\max\{\bar{p}_{i}, b^{2\text{nd}}_{r'_i}\}-\bar{p}_{i}=u_i$ according to allocation rule. So misreporting valuation cannot gain a higher utility for any local seller.

\noindent \textbf{Case 1.3}
For other buyers $i\notin \mathcal{Y}$, her utility is 0. To gain some payoffs, she needs to reveal a larger enough bid and becomes the new winner. However, her purchasing price will be the original highest bid among all buyers which must be larger than her valuation, then her utility will be a negative value.

Above all, we have showed that no buyers can obtain a higher utility by revealing bids different from her valuation. Then, given buyer $i$'s declared valuation $v'_i=v_i$, we prove that for all kinds of buyers, inviting all neighbors (i.e., $r'_i=r_i$) maximizes their utilities. Before that, we can easily find the fact that, as the intended computational action outputs the highest value among all received bids, inviting less neighbors may only decrease her aggregated bid which is used to participate in previous local auctions.

\noindent \textbf{Case 1.4}
If $i$ is the winner when diffusing the auction information to all neighbors, her utility is $v_i-\bar{p}_{i}$. 
For winner who is the highest bidder, $i=z$, she will still win for any $r'_i\neq r_i$.
And for winner who is an ancestor of $z$, she will belong to other buyers if she is not the ancestor of $z$ by inviting less neighbors. 
Otherwise, her utility always keeps unchanged since her selling price may only decrease or remain such that the allocation result will not change and her purchasing price is not related to her neighbors.

\noindent \textbf{Case 1.5}
For a local seller, her utility is $u_i=\max\{\bar{p}_{i}, b^{2\text{nd}}_{r'_i}\}-\bar{p}_{i}$.
For any $r'_i\subseteq r_i$, her aggregated bid may decrease, which will affect the allocation result of previous local auctions. 
If she can still win the item, her selling price $\max\{\bar{p}_{i}, b^{2\text{nd}}_{r'_i}\}$ must be no larger than before since the second-highest bid among her children may decrease when she invites less neighbors, and her purchasing price is independent of her neighbors. So her new utility is always no larger than before.
Otherwise, she may not be able to win the item from previous local auctions, which causes that she belongs to other buyers now and her new utility will be 0.

\noindent \textbf{Case 1.6}
For other buyers $i\notin\mathcal{Y}$ and any $r'_i\subseteq r_i$, she will always be other buyers with utility equaling to 0.

Hence, any buyer $i\in N$ has no incentives to violate the intended information-revelation action $t^M_i$. 

2. For determined $t^M_i, f^M_i$ and other buyers' strategies $s^M_{-i}$, we then show that any buyer has no incentives to violated the intended message-passing action $q^M_i$. As we mentioned in section~\ref{sec:mechanism}, passing a false aggregated bid is equivalent to executing another computational action with different outputs but executing the intended message-passing action, so we only care about reporting the aggregated bid to one or more inviters for the message-passing action here and leave the discuss for passing a false aggregated bid later. For any $q'_i\neq q^M_i$, if a buyer reports nobody, she will have no chance to be on any resale path, thus her utility will always be 0. If she reports her aggregated bid to more than one inviters, the bid may be propagated through more than one paths, and it may raise her purchasing price then reduce her utility.

3. For determined $t^M_i, q^M_i$ and other buyers' strategies $s^M_{-i}$, we prove that executing any other computational action with different outputs ($f'_i\neq f^M_i$) cannot increase the utility for any buyer.

\noindent \textbf{Case 3.1} For the winner $w$, her utility is $u_{w}=v_{w}-\bar{p}_{w}$. If her aggregation algorithm (i.e. her computational action) outputs a lower value than before, she may lose the item so $u'_{w}=0\leq u_{w}$. Otherwise, the allocation result and her purchasing price will not change so her utility keeps unchanged.

\noindent \textbf{Case 3.2} For a local seller, i.e., $\forall y_j\in \mathcal{Y}\setminus\{S, w\}$, her utility is $u_{y_i}=\max\{\bar{p}_{y_i}, b^{2\text{nd}}_{r_{y_i}}\}-\bar{p}_{y_i}$. 
If her chosen algorithm outputs a higher value than before, $f'_i>f^M_i$, it is equivalent to the situation that she decides a larger enough bid $v'_i=f'_i$ and executes the intended computational action. This situation has been discussed in the first point.
If she outputs a lower value than before, the resale path may change to a new one and she may not belongs to the new resale path such that her utility will be zero.

\noindent \textbf{Case 3.3} For any other buyer $i\notin \mathcal{Y}$, her utility is 0 when she executes our intended computational action. To gain some payoffs, she needs to output a larger enough value such that the resale path changes and she becomes the final winner. In this case, her purchasing price must be the highest valuation among the remaining buyers and her utility changes to $u_i=v_i-\bar{p}_{i}\leq0$ where the purchasing price $\bar{p}_{i}$ must be the old highest bid larger than her valuation $v_i$. The new utility is worse than truthfully executing.
Therefore, any buyer cannot gain a higher utility by disobeying the intended computational action.

Taking all together, 
the distributed implementation is ex-post incentive compatible and the intended strategy $s^M$ is an ex-post Nash equilibrium.
\end{proof}

Ex-post IC is commonly achieved in the distributed mechanism design literature and IC is impossible to get. The intuition is that if one agent does not follow the designed computation process, other agents may also change their behaviors to correct the agent’s misbehavior. However, the centralized SRA is IC because the execution is done by the center and for each randomly chosen resale path, the execution is similar to IDM which is proved to be IC in~\cite{li2017mechanism}.

Finally, we can guarantee that our mechanism will not sacrifice the seller's revenue compared to the traditional VCG only among neighbors, which encourages the seller to apply our mechanism.
\begin{proposition}\label{thm:revenue}
The seller's revenue of the connecting-based distributed auction and its centralized reduction mechanism is always no less than that of traditional VCG without diffusion.
\end{proposition}
\begin{proof}
As we decompose the entire sale into a series of local auctions, the seller's revenue only depends on the first local auction hosted by $S$. From Eq.\ref{eq:payment}, the seller's revenue is $b^{2\text{nd}}_{r_S}$ which is the second-highest aggregated bid among all neighbors. Each neighbor of the seller represents a branch to aggregate bids among the subtree rooted at the neighbor, so each neighbor's aggregated bid will always no less than her own bid. Therefore, the second-highest aggregated bid will always no less than the second-highest bid among neighbors which is the revenue under traditional VCG among neighbors without diffusion.
\end{proof}

\subsection{Experimental Results}
At last, since our distributed mechanism does not rely on any existing centralized mechanism, we conduct experiments to show the differences between the centralized reduction of our SRA and others. We choose the IDM mechanism~\cite{li2017mechanism} as a representative to compare with, which is the first proposed IC diffusion mechanism.

\noindent \textbf{Experiment Settings.} We conduct our experiment on the graph shown in Figure~\ref{fig:experiment}. The valuation distribution of each agent is related to her distance to the seller, i.e., the buyer away from the seller has a higher chance to have a higher valuation. This is to demonstrate the characteristics of diffusion auctions more clearly since the goal is to find higher bids in the network. Concretely, we set the valuation distribution of a buyer $i$ with depth (i.e., the length of the shortest path from $S$ to $i$) $1\leq d_i\leq 4$ to be U[$0.1+0.1d_i$, $0.6+0.1d_i$]. We sample $10^4$ instances according to the given distributions and run the centralized reduction of SRA and IDM on these instances respectively. Especially, since the centralized reduction of SRA is a randomized mechanism, we take the average of $10^3$ times running as the estimated expected result. We record the winning probabilities and average utilities over all $10^4$ samples of two mechanisms.

\begin{figure}[htbp]
\centering
    \includegraphics[width=.20\linewidth]{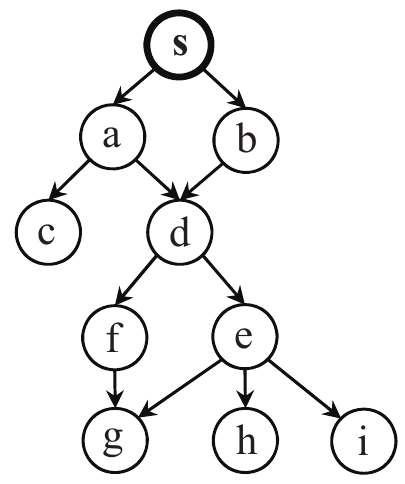}
\caption{The network for the experiment.}
\label{fig:experiment}
\end{figure}

\noindent \textbf{Results and Observations.} We summarize the numerical results of the experiment in Figure~\ref{fig:prob} and Figure~\ref{fig:utility}. In Figure~\ref{fig:prob}, it shows the distribution of winning probabilities of all buyers, and in Figure~\ref{fig:utility}, it shows the expected utilities of all buyers. From these results, we can observe that our mechanism gives more equal chances to win for all buyers. Moreover, it rewards more buyers and especially gives some utilities to buyers near to the seller and those non-critical buyers (e.g., buyer $a$, $b$), who have almost no rewards in the IDM. This will incentivize the buyers near the seller or the non-critical buyers to be more willing to engage in the diffusion process.



\begin{figure}[htbp]
	\centering
	\begin{minipage}{0.495\linewidth}
		\centering
		\includegraphics[width=0.999\linewidth]{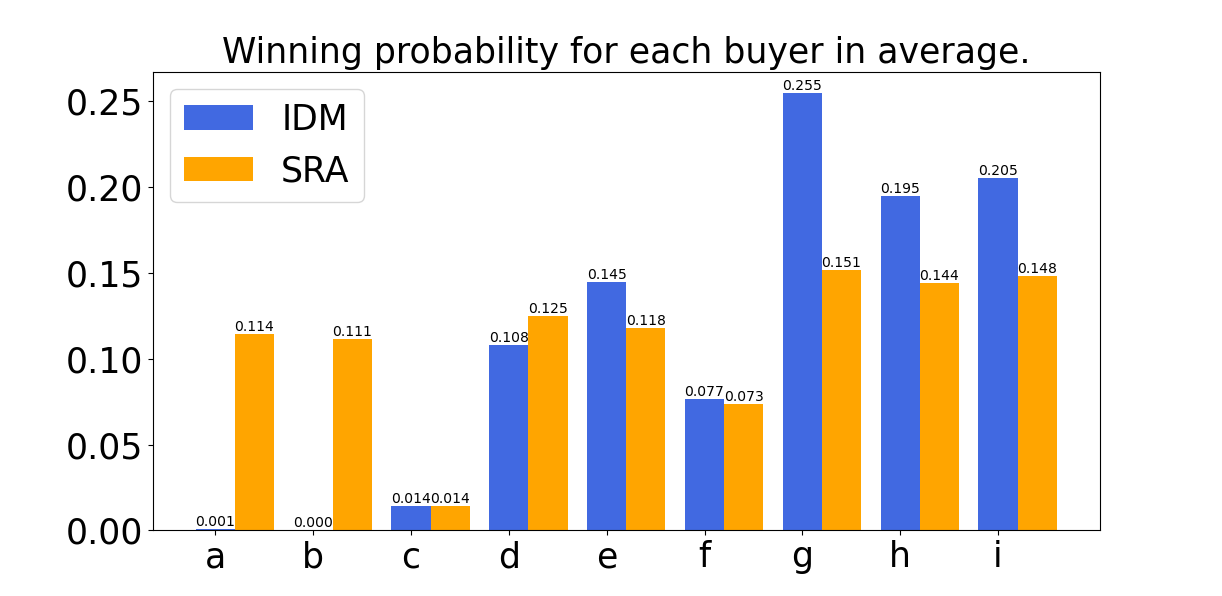}
		\caption{Winning probability for each buyer in average.}
        \label{fig:prob}
	\end{minipage}
	\begin{minipage}{0.495\linewidth}
		\centering
		\includegraphics[width=0.999\linewidth]{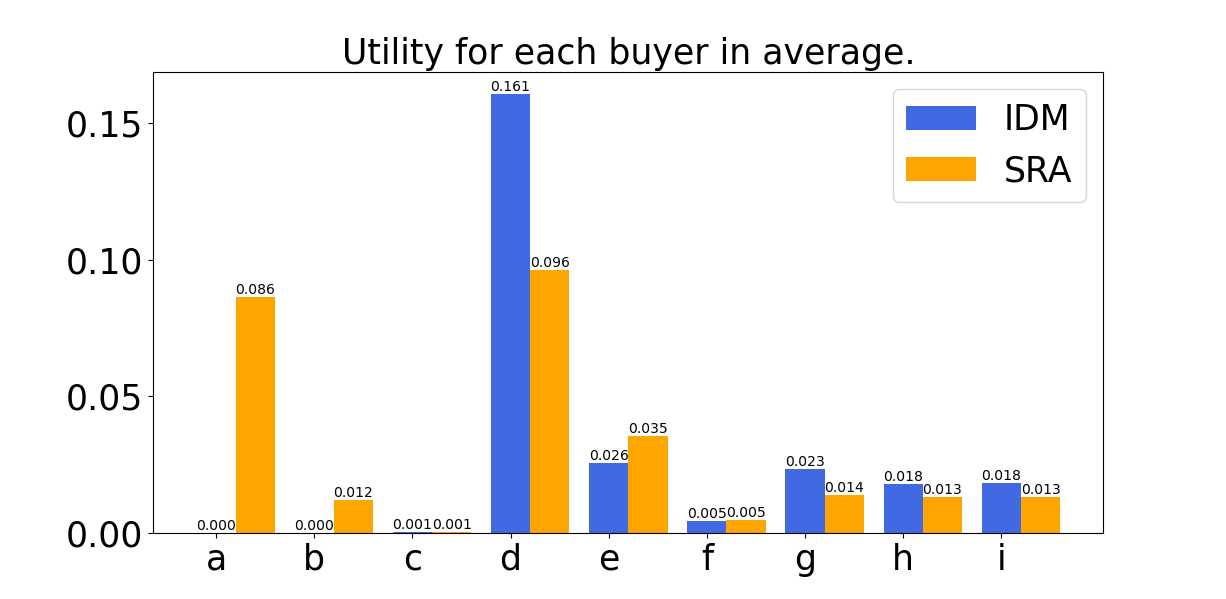}
		\caption{Utility for each buyer in average.}
        \label{fig:utility}
	\end{minipage}
\end{figure}

Therefore, our mechanism can not only distributively computed without any trustworthy centers (which is our core contribution), but also give incentives to more buyers including those non-critical buyers to invite others in the centralized version, which is an interesting and worthwhile extra effects.


\section{Conclusion}\label{sec:conclusion}
In this paper, 
we 
propose the first distributed mechanism in social networks called the Sequential Resale Auction. 
The distributed auction achieves complete decentralization without relying on any trustworthy third party.
We also present the centralized reduction mechanism of our distributed auction to exhibit the extra contribution of our mechanism, which provides a novel way to reward more buyers including those non-critical buyers. 